\documentclass[reprint,aps,twocolumn,pra,showpacs]{revtex4-1}
\usepackage{amssymb}
\usepackage{amsfonts}
\usepackage{amsmath}
\usepackage{graphicx}
\usepackage{bm}
\usepackage{mathrsfs}  

\usepackage[colorlinks=true, citecolor=blue, urlcolor=blue ]{hyperref}

\newcommand{\tr}{\mathrm{tr}}
\newcommand{\rank}{\mathrm{rank}}

\newcommand{\bra}[1]{\left\langle #1 \right|}
\newcommand{\ket}[1]{\left| #1 \right\rangle}

\newtheorem{theorem}{Theorem}[section]
\newtheorem{lemma}[theorem]{Lemma}

\newtheorem{corollary}[theorem]{Corollary}

\newenvironment{proof}[1][Proof]{\begin{trivlist}
\item[\hskip \labelsep {\bfseries #1}]}{\end{trivlist}}
\newenvironment{definition}[1][Definition]{\begin{trivlist}
\item[\hskip \labelsep {\bfseries #1}]}{\end{trivlist}}

\newenvironment{remark}[1][Remark]{\begin{trivlist}
\item[\hskip \labelsep {\bfseries #1}]}{\end{trivlist}}

\newcommand{\qed}{\nobreak \ifvmode \relax \else
      \ifdim\lastskip<1.5em \hskip-\lastskip
      \hskip1.5em plus0em minus0.5em \fi \nobreak
      \vrule height0.75em width0.5em depth0.25em\fi}

\begin{document}
\pacs{03.67.Pp, 05.30.Pr, 71.10.Pm}
\title{Parafermion stabilizer codes}
\author{Utkan G\"{u}ng\"{o}rd\"{u}, Rabindra Nepal, Alexey A. Kovalev}
\affiliation{Department of Physics and Astronomy and Nebraska Center for Materials
and Nanoscience, University of Nebraska, Lincoln, Nebraska 68588,
USA}
\date{\today}
\begin{abstract}
We define and study parafermion stabilizer codes which can be viewed as generalizations of Kitaev's one dimensional model of unpaired Majorana fermions. Parafermion stabilizer codes can protect against low-weight errors acting on a small subset of parafermion modes in analogy to qudit stabilizer codes. Examples of several smallest parafermion stabilizer codes are given. A locality preserving embedding of qudit operators into parafermion operators is established which allows one to map known qudit stabilizer codes to parafermion codes. We also present a local 2D parafermion construction that combines topological protection of Kitaev's toric code with additional protection relying on parity conservation. 
\end{abstract}  
\maketitle

\section{Introduction}
Topologically protected systems are potentially useful for realizations of  fault tolerant elements in a quantum computer~\cite{Kitaev2003,Nayak.Simon.ea:RoMP2008}.
The zero temperature stability of such systems leads to exponential suppression of decoherence induced by local environmental perturbations. On the other hand, the manipulation of the degenerate ground state can be achieved by braiding operations with non-Abelian anyons~\cite{Moore.Seiberg:CiMP1989,Witten:CiMP1989}. 

The Kitaev chain provides an enlightening example of how interactions can result in non-Abelian quasiparticles~\cite{Kitaev2001}. Networks of one dimensional realizations of such quasiparticles can be employed for realizations of quantum gates via braiding operations~\cite{Alicea.Oreg.ea:NP2011,Clarke.Sau.ea:2011}. However, only a non-universal set of quantum gates can be realized with Majorana zero modes. A generalization of Kitaev chain model has been proposed recently where quasiparticles obey parafermion $\mathbb{Z}_D$ algebra as opposed to $\mathbb{Z}_2$ algebra for Majorana zero modes~\cite{Fendley2012}. Many recent publications address possible realizations of parafermion zero modes~\cite{Barkeshli.Qi:PRX2012,Lindner.Berg.ea:PRX2012,Clarke2013,Cheng:2012,Vaezi:2013,Barkeshli.Jian.ea:2013,Barkeshli.Jian.ea:2013b,Hastings.Nayak.ea:2013,Oreg.Sela.ea:2014,Burrello.vanHeck.ea:2013,Mong.Clarke.ea:PRX2014,Klinovaja2014,Tsvelik2014,Ortiz2012,Nussinov2008,Schulz2012,Bullock2007,Vaezi2014,Nigg2014,Bondesan2013,
Motruk2013}. 
The braiding properties 
of parafermion systems have some advantages over 
the 
Majorana modes, while still remaining  non-universal~\citep{Clarke2013,Lindner.Berg.ea:PRX2012,Hastings.Nayak.ea:2013}. However, parafermion systems can be used for obtaining quasiparticles that permit universal quantum computations~\citep{Mong.Clarke.ea:PRX2014}. 

The presence of finite temperature introduces inevitable errors and in principle requires continuous error correction~\cite{Gottesman1997}. `Self correcting' quantum memories are stable at finite temperatures~\cite{Dennis2002,Bacon:2006}; however, they cannot be realized in two dimensions with local interactions~\cite{Bravyi.Terhal:NJoP2009,Landon-Cardinal.Poulin:PRL2013}. Parafermion stabilizer codes considered here can protect against low-weight fermionic errors, i.e. errors that act on a small subset of parafermion modes. The measurement and manipulation schemes required for code implementations have been formulated for Majorana zero modes~\cite{Sau.Tewari.ea:2010,Hassler.Akhmerov.ea:NJoP2010,Jiang.Kane.ea:PRL2011} and should in principle generalize to parafermion zero modes~\cite{Clarke2013}.

In this paper, we address the possibility of active error correction in systems containing a set of parafermion modes as opposed to typical systems containing qubits or qudits. Earlier works on quantum error correction usually addressed the qubit case with a Hilbert space dimension $D=2$~\cite{Shor:1995,Knill.Laflamme:1997,Steane:PRL1996,Gottesman1997}. Error correction on qudits with $D>2$ has also been considered and qudit stabilizer codes have been introduced~\cite{Rains:ITITo1999,Ashikhmin.Knill:ITIT2001,Schlingemann.Werner:2002,GRASSL.BETH.ea:IJoQI2004,Looi.Yu.ea:PRA2008,Hu.Tang.ea:2008,Gheorghiu2010}. The formalism is usually applied to situations in which $D$ is prime or a prime power~\cite{Ashikhmin.Knill:ITIT2001,Ketkar2006,Chen.Zeng.ea:PRA2008} while generalizations to composite $D$ are also possible~\cite{Gheorghiu2014}.   

Parafermion codes can be also interpreted in terms of term-wise commuting Hamiltonians of interacting parafermion zero modes, thus generalizing the Kitaev's one-dimensional (1D) model of unpaired Majorana fermions to $D>2$ case and to arbitrary interactions preserving the commutativity of terms in the Hamiltonian. Of particular interest are the Hamiltonians corresponding to geometrically local interactions on a $d$-dimensional lattice. 
Thus, one can ask similar questions to those posed in Ref.~\cite{Bravyi2010} in relation to Majorana codes, i.e. what is the role  of superselection rules in the finite temperature stability of topological order defined by interacting parafermion modes. Such superselection rules are characteristic to fermionic systems when only interactions with bosonic environments are present. On the other hand, the superselection rule prohibiting parity violating
error operators is not likely to always hold, for instance, when the environment supports
gapless fermionic modes that can couple to the system~\cite{Rainis:may2012,Burnell:dec2013}. Parafermion stabilizer codes can help in such situations by providing protection associated with the code distance of parity violating logical operators. 

The paper is organized as follows. In Section II, we introduce notations and provide background on the theory of qudit stabilizer codes. Here we also discuss the Jordan-Wigner transformation which leads us to introduction of parafermion operators. In Section III, we give formal definition of parafermion stabilizer codes and establish their basic properties. We also discuss the commutativity condition on stabilizer generators, define the code distance, and prove basic results on the dimension of the code space. In Section IV, we present several examples of the smallest parafermion stabilizer codes. In Section V, we construct mappings between qudit stabilizer codes and parafermion stabilizer codes. By employing such mappings, we are able to construct parafermion toric code with adjustable degree of protection against the parity violating errors. Finally, we give our conclusions in Section VI.

\section{Background}

\subsection{Qudits}
Qudits are $D$-dimensional generalizations of qubits, and generally implemented using $D$-level physical systems. One of the well-known generating 
sets for qudit operations is constructed by the generators of the finite \emph{discrete Weyl group} $\mathcal W_D$ that obey the defining relations \cite{weyl1950theory,schwinger2001quantum} 
\begin{align}
X^D = Z^D = \openone, \qquad Z X = \omega X Z.
\label{qcom}
\end{align}
This group is sometimes referred to as discrete Heisenberg group \cite{Ashikhmin.Knill:ITIT2001}, and the generators are sometimes referred to as generalized Pauli matrices \cite{Gheorghiu2014}.
By diagonalizing one of these operators, say $Z$, one obtains the $D$-dimensional representation
\begin{align}
 X = \sum_{j=0}^{D-1} \ket{j+1}\bra{j},\qquad Z = \sum_{j=0}^{D-1} \omega^j \ket{j}\bra{j},
\end{align}
where $\omega = e^{2 \pi i/D}$ and the addition $j+1$ is in mod $D$.
Above and throughout the paper, $\openone$ denotes the identity operator with proper dimensions. Products of $X$ and $Z$ span the Lie algebra $\mathfrak{su}(D)$, hence their linear combinations can generate universal SU($D$) operations.
Operations on multiple qudits are tensor products of the single-qudit operators, hence operators acting on distinct qudits commute.
We will denote an $X$ operator acting on the $j$th site as $X_j$ which is equivalent to an $X$ operator at the $j$th slot of the tensor 
product padded with identity operators: $X_j = \openone \otimes \ldots \otimes X \otimes \ldots \otimes \openone$ (and similar for $Z_j$).

\subsection{Stabilizer codes for qudits}
Stabilizer codes are an important class of quantum error-correcting codes~\cite{Gottesman1997,Gottesman1996} which, under appropriate mapping, can be also thought of as additive classical codes~\cite{Calderbank98quantumerror}. Stabilizer codes utilize a set of commuting operators, called the stabilizer group, for defining the code space. In this section, 
we review the stabilizer formalism for qudits (see e.g. \cite{Gheorghiu2014}).
Let $\mathcal{S}$ be a maximal Abelian subgroup of $\mathcal{W}_D^{\otimes n}$ that does not contain $\omega^j \openone$ ($j \in \mathbb{Z}_D$ and $ j \neq 0$)
and  $C_\mathcal{S}$ be the code subspace of the Hilbert space stabilized by all the elements of $\mathcal S$, i.e. $S_i \ket{\psi}=\ket{\psi}$ 
$\forall S_i \in \mathcal{S}$ and $\ket{\psi} \in C_{\mathcal S}$,
then $\mathcal{S}$ is called the  stabilizer group and it is generally denoted by its generating set $\mathcal S = \langle S_1, S_2, ....., S_k \rangle$. 

Since the stabilizer group $\mathcal{S}$ is an Abelian group, its elements must commute with each other by definition.
The commutativity condition of its generators depends upon the particular case of $\mathcal W_D^{\otimes n}$ at hand. 
Two arbitrary elements of $\mathcal W_D^{\otimes n}$, $G = \omega^{\lambda}X^{\bm{u}}Z^{\bm{v}}$ and 
$G' = \omega^{\lambda'}X^{\bm{u'}}Z^{\bm{v'}}$
where $X^{\bm{u}} =X_1^{u_1}X_2^{u_2}....X_n^{u_n}$, $Z^{\bm{v}} =Z_1^{v_1}Z_2^{v_2}....Z_n^{v_n}$ (and similar for $G'$) will commute iff
\begin{align}
\bm{u}.\bm{v'} = \bm{v}.\bm{u'}   \mod{D}
\end{align}
is satisfied \cite{Gheorghiu2014}.

The support  of a Weyl operator $w \in \mathcal{W}_D^{\otimes n}$, denoted as $\text{Supp}(w)$, is defined as the set of qudits on which it acts non-trivially. The cardinality of the support, $|\text{Supp}(w)|$, is called the weight of the
 operator $w$, also denoted as $|w|$. The set of all Weyl operators in $\mathcal{W}_D^{\otimes n}$ that commute with all the elements of $\mathcal S$ is called the \emph{centralizer} of 
$\mathcal S$ and is denoted as $\mathcal{C}(S)$.

For prime $D$, a stabilizer group with $n-k$ independent generators implies that the corresponding centralizer is generated by $n+k$ generators. 
The logical operators $\{\bar{X},\bar{Z} \}$ of a stabilizer code $\mathcal S$ are the elements of $\mathcal{C}(\mathcal S)$ that are not in $\mathcal S$. 

The robustness of a quantum code can be measured by how far two encoded states are apart, which is quantified through the notion 
of \emph{distance}. The weight of the logical operators imply the separation of the encoded states. Therefore, the distance of a stabilizer 
code is defined  as,
\begin{equation}
d = \min_{L_i \in \mathcal C(\mathcal S)\setminus \mathcal S}|L_i|.
\end{equation}
The longer is the code distance the better protection the code provides. A
code of distance $d$ can detect any error of weight up to $d-1$, and correct up to $\lfloor d/2\rfloor$.
A quantum error-correcting code that encodes $n$ physical qudits into $k$ logical qudits with distance $d$ is denoted as $[[n,k,d]]_D$.

\subsection{Parafermion operators}

Parafermion operators can be obtained by the Jordan-Wigner transformation of the $D$-state spin operators $\{X_j,Z_j\} \in \mathcal W_D^{\otimes n}$ as,
\begin{align}
\gamma_{2j-1} &= \left(\prod_{k=1}^{j-1}X_{k}\right)Z_j, \nonumber \\ 
\gamma_{2j}  &= \omega^{(d-1)/2} \left(\prod_{k=1}^{j-1}X_{k}\right) Z_j X_{j}.
\label{wj}
\end{align}
which is a mapping of $n$ \emph{local} spin operators into $2n$ \emph{non-local} parafermion operators, therefore, the total number of parafermion modes  is always even.
Parafermion operators $\gamma_j$ obey the following relations:
\begin{align}
\gamma_j^d = \openone, \quad
\gamma_j \gamma_k = \omega \gamma_k \gamma_j \quad (j<k, \quad \omega = e^{i 2\pi/D})
\label{comm}
\end{align}
Special case with $D=2$ gives us the anti-commuting self-adjoint Majorana fermions.

Realizations of parafermion zero modes corresponding to Eq.~(\ref{comm}) have been suggested. In such realizations, the localized state is described by parafermion operator which commutes with the corresponding Hamiltonian and changes the parity of $\mathbb Z_D$ charge by 1~\cite{Fendley2012}.
They are non-Abelian anyons and can be used for realizations of fault-tolerant topological quantum gates.

There are recent proposals to construct
solid state systems that accommodate parafermion zero modes. Realizations employing  exotic fractional quantum Hall (FQH) states and quantum nanowires have been proposed~\cite{Barkeshli.Qi:PRX2012,Lindner.Berg.ea:PRX2012,Clarke2013,Cheng:2012,Vaezi:2013,Barkeshli.Jian.ea:2013,Barkeshli.Jian.ea:2013b,Hastings.Nayak.ea:2013,Oreg.Sela.ea:2014,Burrello.vanHeck.ea:2013,Mong.Clarke.ea:PRX2014,Klinovaja2014,Tsvelik2014}.

\section{Parafermion stabilizer codes}
\subsection{The group PF($D,2n$)}
We shall call the group generated by the single-mode operators $\gamma_j$ given in Eq. (\ref{comm})
the parafermion group PF($D,2n$). Arbitrary elements of PF($D,2n$) can be written as $\omega^\lambda \gamma^{\bm{\alpha}}$ where $ \lambda \in \mathbb{Z}_D$ and
 \begin{equation}
\gamma^{\bm{\alpha}} = \gamma_1^{\alpha_1}\ldots\gamma_{2n}^{\alpha_{2n}}
\label{rel}
\end{equation}
with $\bm{\alpha} = (\alpha_1, \ldots, \alpha_{2n}) \in \mathbb{Z}_D^{2n}$ and by convention the terms are  arranged in increasing order in their indices. The ordered set of non-zero elements in $\bm{\alpha}$ is called the support of $\gamma^{\bm{\alpha}}$, or $\text{Supp}(\gamma^{\bm{\alpha}})$. We define the weight of $\gamma^{\bm{\alpha}}$ as the number of non-zero entries in $\bm{\alpha}$, denoted as $|\text{Supp}(\gamma^{\bm{\alpha}})|$ or simply $|\gamma^{\bm{\alpha}}|$.

A parafermion operator $\omega^\lambda \gamma^{\bm{\alpha}} \in \text{PF}(D,2n)$ will preserve parity iff
\begin{align}
\sum_{i=1}^{2n} \alpha_i = 0 \mod D.
\label{eq:parity}
\end{align}

One can generalize Eq. (\ref{comm}) to obtain $\gamma_i^m \gamma_j^n = \omega^{mn} \gamma_j^n \gamma_i^m$ for $i<j$. Using this, it can be shown that two parafermion operators $\gamma^{\bm{\alpha}}$ and $\gamma^{\bm{\beta}}$ commute iff
\begin{equation}
\bm{\alpha} \Lambda \bm{\beta}^T = 0 \mod D
\label{parafermions-commutation}
\end{equation}
is satisfied where
$\Lambda$ is a $2n \times 2n$ anti-symmetric matrix $\Lambda_{ij} = \text{sgn}(j-i)$ or explicitly
\begin{eqnarray}
\Lambda =
\begin{pmatrix}
0&1&1\ \dots\  & 1\\
-1&0&1\ \dots\   & 1\\
-1&-1&0\ \dots    &1\\
\vdots&\vdots&\vdots&\vdots \\
-1&-1&-1\ \dots\ & 0
\end{pmatrix}.
\end{eqnarray}

In particular, when the index of the last non-zero entry in ${\bm \alpha}$ is smaller than the index of the first non-zero entry in ${\bm \beta}$, the commutativity condition Eq. (\ref{parafermions-commutation}) is reduced to
\begin{align}
\left(\sum_j \alpha_j\right)\left(\sum_j \beta_j\right) = 0 \mod D.
\end{align}

The parity-conservation condition for a parafermion operator can also be expressed in terms of the $\mathbb Z_D$ charge operator
\begin{align}
Q = \prod_{j=1}^n \gamma_{2j-1}^\dagger \gamma_{2j}.
\end{align}
For any $\gamma^{\bm{\alpha}} \subset \text{PF}(D,2n)$
\begin{align}
\gamma^{\bm{\alpha}} Q  = \omega^p Q \gamma^{\bm{\alpha}}, \qquad p = \sum_{i=1}^{2n} \alpha_i \mod D,
\end{align}
where $p$ is the $\mathbb Z_D$ charge of $\gamma^{\bm{\alpha}}$, thus the parity-conservation condition can also be written as $[\gamma^{\bm{\alpha}}, Q]=0$.

Since Majorana zero modes correspond to the $D=2$ case, evidently we have $\text{PF}(2,2n) \cong \text{Maj}(2n)$.

\subsection{Stabilizer groups in PF($D,2n$)}
It is not generally possible to map parafermion operators in PF($D,2n$) onto qudit operators in $\mathcal W_D^{\otimes k}$ due to the non-locality of parafermion operators. The tensor product structure of $k$-qudit operators in $\mathcal W_D^{\otimes k}$ guarantees that operators acting on different sites commute, whereas parafermion operators fail to commute for all distinct sites. Nevertheless, even though a one-to-one mapping between a single-mode parafermion operator and a qudit operator is impossible, it is indeed possible to map multiple parafermion modes onto multiple qudits at once (see subsection \ref{sec:parafermion-stabilizer-codes}) or to map multiple parafermion modes onto a local single-qudit in a consistent way (see Section \ref{sec:qudits-from-parafermions}). Indeed, as we observe in the next section, PF($D,2n$) proves to be rich group with many non-trivial Abelian subgroups.

\begin{definition}
Parafermion stabilizer codes $C_{\mathcal S_{PF}}$, similar to qudit stabilizer codes, are completely determined by their corresponding stabilizer group, which in our case is $\mathcal S_{PF}\subseteq \text{PF}(D,2n)$. We list the defining properties of parafermion stabilizer codes as:
\begin{itemize}
\item Elements of $\mathcal S_{PF}$ are parity-preserving operators.
\item $\mathcal S_{PF}$ is an Abelian group not containing $\omega^ j \openone$ where $j \in \mathbb Z_D$ and $j \neq 0$.
\end{itemize}
Whether these conditions hold for a given parafermion stabilizer code or not can be verified using Eqs. (\ref{eq:parity}) and (\ref{parafermions-commutation}) respectively.
\end{definition}

The set of all parafermion operators in $\text{PF}(D,2n)$ which commute with all the elements of $\mathcal S_{PF}$ is called the \emph{centralizer} of $\mathcal S_{PF}$ and is denoted as 
$\mathcal{C}(\mathcal S_{PF})$. The set of logical operators $\mathcal L(\mathcal S_{PF})$ encoding $k$ qudits of a parafermion code $\mathcal S_{PF}$ are the elements of $\mathcal{C}(\mathcal S_{PF})$ that are not in $\mathcal S_{PF}$, that is $\mathcal L(\mathcal S_{PF}) = \mathcal C(\mathcal S_{PF})\setminus \mathcal S_{PF}$. When $D$ is a prime number, the order of the generating set (excluding the identity operator) of $\mathcal S_{PF}$ is $n-k$ and the centralizer is generated by $n+k$ generators.

When writing the generating sets explicitly, we will omit the phase factors $\omega^l$ ($l \in \mathbb Z_D$) for all generators for brevity throughout the paper, however, one should keep in mind that such phase factors are in general required in order to satisfy the second defining property of parafermion codes listed above.

The codespace of a parafermion stabilizer code $\mathcal S_{PF}$ is the subspace that is invariant under the action of all the elements of $\mathcal S_{PF}$.

The distance $d$ of a parafermion code is given by the minimum weight of its logical operators, 
\begin{equation}
d = \min_{\mathcal{\gamma^{\bm \alpha}}\in \mathcal L(\mathcal S_{PF})}|\gamma^{\bm \alpha}|.
\end{equation}
We denote a parafermion stabilizer code that encodes $2n$ parafermion modes into $k$ logical qudits with distance $d$ as $[[2n,k,d]]_D$. A parafermion stabilizer code of distance $d$ can detect any parafermion error of weight up to $d-1$, and correct up to $\lfloor d/2\rfloor$ in analogy to qudit codes.
However, it should be noted that similar to Majorana fermion codes \cite{Bravyi2010} the robustness of parafermion codes is not solely determined by the code distance $d$: when some of the logical operators have non-zero parity, the conservation of parafermion parity will offer additional protection, that is, a subspace of the codespace will be protected against such errors. Following Ref.~\cite{Bravyi2010} we introduce an additional parameter $l_\text{con}$ defined as the minimum diameter of a region that can
support a parity conserving logical operator:
\begin{equation}
l_\text{con} = \min_{\substack{
\mathscr{ \gamma^{\bm \alpha}}\in \mathcal L(\mathcal S_{PF})  \\
\sum_i\alpha_i=0\mod D
 }}  \text{diam}[\text{Supp}(\gamma^{\bm \alpha})],
\end{equation}
which can be used in order to measure the degree of protection relying on the superselection rules.




What can be said about the order of $\mathcal S_{PF}$? Below, we adapt the theorem and proof given by Gheorghiu \cite{Gheorghiu2014} to parafermion stabilizer codes.
\begin{theorem}[Gheorghiu]
Let $\mathcal S_{PF}$ be a parafermion stabilizer code in PF($D,2n$) where $D$ is allowed to be composite, let $|\mathcal S_{PF}|$ denote the order of $\mathcal S_{PF}$ and let $|C_{\mathcal S_{PF}}|$ be the dimension of codespace. Then the following equation holds:
\begin{align}
 |C_{\mathcal S_{PF}}| |\mathcal S_{PF}| = D^n.
\label{eq:dimensionality} 
\end{align}
\end{theorem}

\begin{proof}
The operator
 \begin{align}
  P = \frac{1}{|\mathcal S_{PF}|} \sum_{j_=1}^{|\mathcal S_{PF}|} S_j.
 \end{align}
 is a projection operator satisfying $P^2 = P = P^{\dagger}$. Clearly, for any $\ket{\psi_j} \in C_{\mathcal S_{PF}}$, $P \ket{\psi_j} = \ket{\psi_j}$ holds. Thus the subspace $\mathcal W$ which $P$ projects onto includes $C_{\mathcal S_{PF}}$, or $C_{\mathcal S_{PF}} \subseteq \mathcal W$.
 
 Next we show that this relation holds the other way around. Let $\ket{\phi}$ be an arbitrary element of $\mathcal W$ (thus $P \ket{\phi} = \ket{\phi}$) and $S_k$ be an arbitrary element of $\mathcal S_{PF}$. Since $S_k P = P$ for all $k$, we obtain $S_k (P \ket{\phi}) = P \ket{\phi}$, meaning all $\ket{\phi} \in \mathcal W$ is stabilized by $\mathcal S_{PF}$ or $\mathcal W  \subseteq C_{\mathcal S_{PF}}$, leading us to the conclusion that $\mathcal W = C_{\mathcal S_{PF}}$. The dimension of the codespace is then given as $\tr(P)$. Since $\mathcal S_{PF}$ is an Abelian group and the trace condition $\tr(\gamma^{\bm{\alpha}})=0$ when $\gamma^{\bm{\alpha}} \neq \openone$ and $\tr(\openone)=D^n$ for $\gamma^{\bm{\alpha}}, \openone \in \text{PF}(D,2n)$ holds, we arrive at the result
 \begin{align}
  \tr(P) = |C_{\mathcal S_{PF}}| = \frac{1}{|\mathcal S_{PF}|} D^n.
 \end{align}

\end{proof}

\begin{corollary}
When $D$ is a prime power $p^l$, $|C_{\mathcal S_{PF}}|=p^{lk}$ and $|\mathcal S_{PF}|=p^r$ with $r=l(n-k)$ (we refer to \cite{Ashikhmin.Knill:ITIT2001} for a detailed derivation).
\end{corollary}

In later sections, we will also use a matrix form of the stabilizer code $\mathcal S_{PF} = \langle S_1, \ldots, S_l \rangle = \langle \gamma^{\bm{\alpha}_1}, \ldots, \gamma^{\bm{\alpha}_l} \rangle$ whose rows are given by $\bm{\alpha}_i$, that is
\begin{align}
 S_{PF} = \begin{pmatrix}
      \bm{\alpha}_1 \\
      \vdots \\
      \bm{\alpha}_l
     \end{pmatrix}.
\end{align}
The same construction is also extended for the logical operators, yielding the matrix $L_{PF}$.
Since $\mathcal S_{PF}$ is an Abelian group, due to Eq. (\ref{parafermions-commutation}), we have $S_{PF} \Lambda S_{PF}^T = 0 \mod D$. The logical operator matrix $L_{PF}$ on the other hand obeys the relations $L_{PF} \Lambda S_{PF}^T = 0$ and $L_{PF} \Lambda L_{PF}^T \neq 0$ in $\mod D$.

\section{Examples of parafermion stabilizer codes}
\subsection{3-state quantum clock model}
We present a simple example of parafermion code starting from a 3-state quantum clock model Hamiltonian (for $h=0$):   
\begin{equation}
H_3 = -J \sum_{j=1}^{n-1}(Z_{j}^\dagger Z_{j+1} + Z_{j+1}^\dagger Z_{j}).
\label{ising3}
\end{equation}
By employing the Jordan-Wigner transformation, this Hamiltonian can be rewritten in terms of parafermion operators in the following form:
\begin{equation}
H = iJ \sum_{j=1}^{n-1} (\gamma_{2j}^{\dagger} \gamma_{2j+1} - \gamma_{2j+1}^{\dagger} \gamma_{2j} ),
\label{ising}
\end{equation}
which is known as the Fendley \cite{Fendley2012} generalization of Kitaev chain model.
For $D=2$, Eq.~(\ref{ising3}) reduces to familiar Ising model with $h = 0$.

We form the corresponding  stabilizer group taking individual terms of the Hamiltonian for each value of $j$ as,
\begin{align}
\langle i\gamma_2^\dagger \gamma_3, -i\gamma_3^\dagger \gamma_2, \ldots,i\gamma_{2n-2}^\dagger \gamma_{2n-1},-i\gamma_{2n-1}^\dagger \gamma_{2n-2}\rangle. 
\end{align}
Logical operators of the code can be chosen as $\bar{Z}=\gamma_1$ and $\bar{X}=\gamma_{2n}$. Then the distance of the code is $d=1$. 
But these logical operators are not parity-preserving, we can combine them as  $\gamma_1^\dagger \gamma_{2n}$ and $\gamma_1 \gamma_{2n}^\dagger$
to obtain parity-preserving logical operators. Even though this code does not provide protection against parity violating errors, in the absence of such errors the code protection can be described by the diameter of even logical operators, i.e., $l_\text{con}=2n$.   

\subsection{Minimal parafermion stabilizer codes}
\label{sec:parafermion-stabilizer-codes}
Quantum error-correcting schemes come at the expense of introducing additional qudits in order to protect information encoded into quantum states. The ratio of the number of encoded qudits $k$ (whose state can be restored after decoherence) to the number of underlying physical qudits $n$ is called encoding rate $r=k/n$. The relative distance is defined as $\delta=d/n$. Codes with higher encoding rate $r$ and relative distance are preferable and it is known that both $\delta$ and $r$ can be finite for a particular code family \cite{Calderbank.Shor:1996}. In this section, we discuss the minimal stabilizer codes encoding $k=1$ qudit and try to find codes with the best encoding rate $r$ for the minimal non-trivial distance $d=3$ for prime $D$.

Using exhaustive search, we find that for $D=3$ the smallest non-trivial code requires 8 parafermion modes and results in a $[[8,1,3]]_3$ parafermion stabilizer code:
\begin{align}
\mathcal S_{PF} =& \langle \gamma_1^\dagger \gamma_2 \gamma_4^\dagger \gamma_6, \gamma_2^\dagger \gamma_3 \gamma_5^\dagger \gamma_7, \gamma_3^\dagger \gamma_4 \gamma_6^\dagger \gamma_8 \rangle, \nonumber \\
\mathcal L(\mathcal S_{PF}) = &\langle\gamma_1^\dagger \gamma_2 \gamma_3 \gamma_7, \gamma_2^\dagger \gamma_3^\dagger \gamma_6\rangle
\label{eq:stabilizer-d=3}.
\end{align}
The logical operators generate $\mathcal W_3$, encoding 8 parafermion modes into a single logical qutrit.

Realizations of $D=6$ parafermion zero modes have been proposed recently \cite{Clarke2013}, making this case particularly interesting. Because $D=6$ is not a prime or prime power, the original construction for qudit stabilizer codes \cite{Ashikhmin.Knill:ITIT2001} is not directly applicable. We will instead ``double'' the $D=3$ code given above by squaring all the generators. However, this is a mapping onto a larger space and we need to take care of the additional operators that commute with the new stabilizer generators. The full set of generators for $D=6$ thus becomes
\begin{align}
\mathcal S_{PF} =& \langle \gamma_1^3 \gamma_2^3, \gamma_3^3 \gamma_4^3, \gamma_5^3 \gamma_6^3, \gamma_7^3 \gamma_8^3, \nonumber \\ 
& 
(\gamma_1^\dagger \gamma_2 \gamma_4^\dagger \gamma_6)^2, (\gamma_2^\dagger \gamma_3 \gamma_5^\dagger \gamma_7)^2, (\gamma_3^\dagger \gamma_4 \gamma_6^\dagger \gamma_8)^2
\rangle, \nonumber \\
\mathcal L(\mathcal S_{PF}) = &\langle (\gamma_1^\dagger \gamma_2 \gamma_3 \gamma_7)^2, (\gamma_2^\dagger \gamma_3^\dagger \gamma_6)^2 \rangle.
\end{align}
Since these logical operators behave like $X^2$ and $Z^2$ for $D=6$ qudits, the code above essentially encodes a qutrit using $2n=8$ parafermion zero modes. We also note that this code may not have the best encoding rate for $D=6$.

However, the minimal number of modes depends on $D$. For the case of $D=7$, there exists $[[6,1,3]]_7$ code that requires only $6$ modes:
\begin{align}
\mathcal S_{PF} =& \langle \gamma_1 \gamma_2 \gamma_5^5, \gamma_1 \gamma_4^5 \gamma_6 \rangle, \nonumber \\
\mathcal L(\mathcal S_{PF}) = &\langle \gamma_1^3 \gamma_2^6 \gamma_6, \gamma_1^2 \gamma_2^5 \gamma_3\rangle.
\end{align}
This indicates that there is a minimal $D$ for which the encoding rate is optimal \footnote{Exhaustive search takes exponential time in $D$, thus we were unable to examine $D>7$ cases and determine the optimal $D$. A better algorithm may allow determining this value.}.

\section{Mappings between qudits and parafermion modes}
\label{sec:qudits-from-parafermions}
\subsection{Mappings to and from parafermion codes}
There is an established literature on stabilizer codes for qudits when $D$ is prime or a prime power
\cite{Viyuela2012,Anwar2012}. Recently, some properties of qudit stabilizer codes for non-prime case has been discussed in \cite{Gheorghiu2014}. An isomorphism between 
multi-qudit and multi-parafermion mode operators will let us construct parafermion stabilizer codes based on qudit codes. 
In this section, we establish such an isomorphism by mapping four parafermion modes to a single qudit.

\begin{remark} Let $\tilde X_j$ and $\tilde Z_j$ ($j = 1 \ldots k$) denote the generating operators of $\mathcal W_D^{\otimes k}$ embedded into PF($D,2n$), encoding $k$ qudits into $2n$ parafermion modes. Such an embedding has these properties:
\begin{itemize}
\item Logical qudit operators $\{\tilde X_j, \tilde Z_j\}$ obey Eq. (\ref{qcom}), that is, they generate the embedded Weyl group $\mathcal W_D^{\otimes k} \subseteq \text{PF}(D,2n)$.
\item Logical qudits operators for different sites commute ($[\tilde X_i, \tilde X_j] = [\tilde Z_i, \tilde Z_j] = [\tilde X_i, \tilde Z_j] = 0$ when $i\neq j$).
\item The embedding of $\mathcal W_D^{\otimes k}$ into the larger group PF($D,2n$) may require additional parafermion operators $\{\tilde Q_j^{(i)}\}$ that commute with the original qudit stabilizer group $\mathcal S$ or its corresponding logical operators $\mathcal L(\mathcal S_{PF})$. Such operators must be included in the parafermion stabilizer group $\mathcal S_{PF}$ and hence must preserve parity (an example is given in Eq. (\ref{eq:qudit-isomorphism}) below).
\end{itemize}
\end{remark}
In turns out that the minimum number of parafermion modes required for such an embedding is four, that is four parafermion modes will map to a single qudit. This mapping leads to the following lemma.
\begin{lemma}
Every $[[n,k,d]]_D$ stabilizer code can be mapped onto a $[[4n,k,2d]]_D$ parafermion stabilizer code, encoding 4 parafermion modes into a single qudit.
\end{lemma}

\begin{proof}

Let us define the operators
\begin{align}
\tilde Z_{j+1} &= \gamma_{1+4j}^\dagger \gamma_{2+4j}, \qquad \tilde X_{j+1} = \gamma_{1+4j}^\dagger \gamma_{3+4j}\nonumber \\
\tilde Q_{j+1}&=\gamma_{1+4j}^\dagger \gamma_{2+4j} \gamma_{3+4j}^\dagger \gamma_{4+4j}
\label{eq:qudit-isomorphism}
\end{align}
It is straightforward to show that $\langle\tilde X_j,\tilde Z_j\rangle$ generate the embedded Weyl group $\mathcal W_D^{\otimes k} \subseteq \text{PF}(D,2n)$ (that is, $\tilde Z_{i} \tilde X_{j} = \omega \tilde X_{j} \tilde Z_{i} \delta_{ij}$ and $\tilde X_{j}^D = \tilde Z_j^D = \openone$) and are parity-preserving.
We can treat $\mathcal L(\mathcal S_{PF}) = \langle\tilde X_j,\tilde Z_j\rangle$ as the logical operators of a stabilizer group $\mathcal S_{PF} = \langle \tilde Q_j \rangle$. This makes the purpose of the additional fourth mode (which does not appear in the logical operators) clear: without it, the stabilizer group would include a non-parity-preserving operator.
Finally, since every Weyl operator is mapped to a parafermion operator with two modes, the distance of the new code is $2d$.
\end{proof}

This mapping allow us to construct families of parafermion stabilizer codes from known families of qudit stabilizer codes. In particular, one can map the qudit toric codes \cite{Viyuela2012} (and their generalizations \cite{Tillich2009,Kovalev.Pryadko:2013}) to the corresponding parafermion code. The advantage of this mapping is that a local stabilizer generator in $d$-dimensional lattice will map to a local parafermion operator. The disadvantage is that all logical operators preserve parity, thus there is no additional protection associated with the presence of parity violating logical operators.   

It turns out that we can do a similar mapping in the opposite direction  albeit without preserving the locality of stabilizer generators.

\begin{lemma}[Doubling]
Any  parafermion stabilizer code with parameters $[[2n,k,d]]_D$ and stabilizer group $\mathcal S_{PF}$ can generate a $[[2n,2k,d']]_D$ qudit CSS code.
\end{lemma}
\begin{proof}
Consider the check matrix
\begin{align}
S_{CSS} = \begin{pmatrix}
  S_{PF} \Lambda & 0 \\
  0 & S_{PF} \\
 \end{pmatrix}.
\end{align}
For a parafermion code, $k = n - \rank(S_{PF})$ whereas for the CSS code $k' = 2n - 2 \times \rank(S_{PF}) = 2k$ ($\Lambda$ is full-rank matrix). Hence $S_{CSS}$ is the check matrix of a $[[2n,2k,d']]$ CSS code. The corresponding logical operator matrices $L_{PF}$ and $L_{PF} \Lambda$, behave like $X$- and $Z$-type logical qudit operators.

We note that this construction is a proper generalization of the doubling lemma described in \cite{Bravyi2010} which maps a Majorana fermion code to weakly self-dual CSS code. Unfortunately, for $D>2$ this mapping becomes non-local, i.e., a local qudit operator will generally map to a non-local parafermion operator.
\end{proof}

\subsection{Parafermion toric code with parity violating logical operators}
In this section, we construct parafermion analog of Kitaev's toric code \cite{Kitaev2003} for qudits \cite{Viyuela2012}. The toric code is a stabilizer code defined on a $a \times b$ lattice on the surface of a torus. A portion of the lattice is depicted in Fig.~\ref{fig:toric} where each dot represents a single qudit (hence, there are $2 a b$ qudits overall).

\begin{figure}[!htb]
\centering
\includegraphics[scale=0.9]{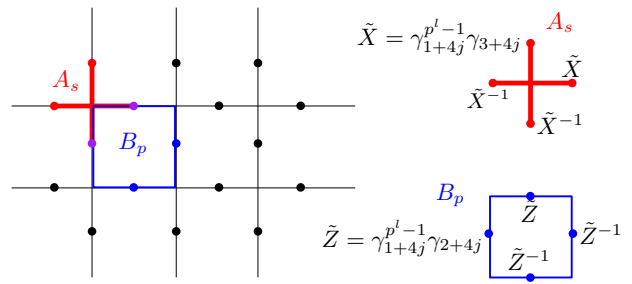}
\caption{A portion of the lattice place on torus where each dot represents four parafermion modes (the index $j \geq 0$ uniquely denotes the lattice point). On the right, parafermion star and plaquette operators $A_s$ and $B_p$ are given in detail ($p^l$ is a prime power, further details are given in Eq. (\ref{eq:toric-mapping}) ) (color online).}
\label{fig:toric}
\end{figure}

Let $D=p^{2l}$ where $p$ is a prime number and $l \in \mathbb Z^+$. The operators
\begin{align}
\tilde Z_{j+1} =& \gamma_{1+4j}^{p^l-1} \gamma_{2+4j}, \qquad \tilde X_{j+1} = \gamma_{1+4j}^{p^l-1} \gamma_{3+4j}, \nonumber \\
\tilde Q_{j+1} =& \gamma_{1+4j}^\dagger \gamma_{2+4j}^\dagger \gamma_{3+4j} \gamma_{4+4j}.
\label{eq:toric-mapping}
\end{align}
define a mapping of four parafermion modes onto a single qudit via the one-qudit stabilizer group $\mathcal S_{PF}=\langle\tilde Q_j\rangle$ and its corresponding logical operators $\mathcal L(\mathcal S_{PF})=  \langle\tilde X_j, \tilde Z_j\rangle$.

Consider the operators defined on a star-shaped and plaquette-shaped portions of the lattice:
\begin{align}
A_s = \prod_{j \in \text{star}(s)} \tilde X_j^{a_j}, \quad B_p = \prod_{j \in \text{plaquette}(p)} \tilde Z_j^{b_j},
\label{eq:toric}
\end{align}
where $a_j$ and $b_j$ are $\pm 1$, specified on the right side of Fig.~\ref{fig:toric}.
In general, $A_s$ and $B_p$ either do not overlap or overlap at two sites. One can easily verify that the construction given in Eq. (\ref{eq:toric}) ensures that the commutator $[A_s, B_p]$ vanishes in both cases. We also note that both $A_s$ and $B_p$ are parity-conserving operators. The set of all $A_s$ and $B_p$ forms a stabilizer group.

Due to the fact that the lattice is defined on the surface of a torus, the lattice is periodic in both dimensions, leading to the result
\begin{align}
\prod_s A_s = \openone, \quad \prod_p B_p = \openone.
\end{align}
This implies $|\mathcal S| = 2 (a b-1)$, and using Eq. (\ref{eq:dimensionality}), we find that $k=2$.
The logical operators ${\mathcal X_l, \mathcal Z_l}$ ($l=1,2$) are horizontal and vertical loops along the lattice, as given in Fig.~\ref{fig:toric-logical}. Since these loops go all the way through the torus, they commute with the stabilizer generators $A_s$ and $B_p$ at all sites.

\begin{figure}[!htb]
\centering
\includegraphics[scale=0.9]{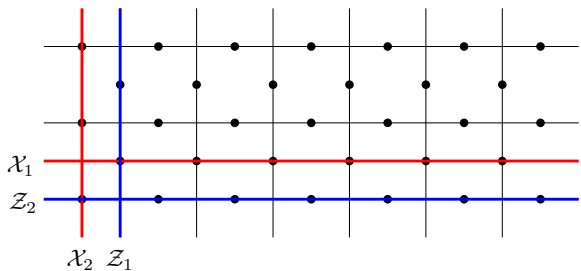}
\caption{Loops corresponding to the logical operators of the toric code (color online).}
\label{fig:toric-logical}
\end{figure}

 We note that the parity (charge) associated with operators is $p^l \neq 0 \mod D$ \footnote{The presence of parity-violating operators does not prevent quantum computation. Kitaev chain contains parity-violating operators as well, nevertheless a topological qubit can be defined by using 4 Majorana edge modes or a pair of topological regions \cite{Fu2009,Lutchyn2010,Drummond2014}. Storage and manipulation of information takes place in the code space corresponding to a given parity sector of the Hilbert space.}. Hence, the parity of the horizontal (vertical) logical operators of the parafermion toric code is $a \times p^l$ ($b \times p^l$) mod $D$. By tuning $a$ and $b$ we can ensure that one of the logical operators will violate parity (that is, $p^l$ divides $a$  but does not divide $b$). The choice of the smallest $b$ would correspond to the absence of parity violating errors. In general, $b$ can be tuned depending on the probability of parity violating errors. Therefore, this code construction combines 
topological protection of Kitaev's toric code with additional protection relying on suppression of parity violating errors.

\section{Conclusion}
We have introduced stabilizer codes in which parafermion zero modes represent the constructing blocks as opposed to qudit stabilizer codes. Our work generalizes earlier constructions based on Majorana zero modes \cite{Bravyi2010}. While it is in general possible to start with a stabilizer code for qudits and use it with parafermion zero modes through the mapping given in Eq. (\ref{eq:qudit-isomorphism}) which utilizes the embedding $\mathcal W_D^{\otimes n} \subset \text{PF}(D,4n)$, we find that there are more efficient codes in $\text{PF}(D,2n)$ requiring less number of parafermion modes as we have exemplified in Section \ref{sec:parafermion-stabilizer-codes}. These results also show that the parafermions can achive better encoding rate than Majorana fermions. We have also shown that using a similar embedding with qudit toric code it is possible to construct a code protecting parafermion modes against parity violating errors where the degree of protection (i.e. distance) can be adjusted. A similar 
construction has been introduced for color codes using Majorana zero modes \cite{Bravyi2010}. 

Parafermion stabilizer codes can be used for constructing Hamiltonians in which commuting terms correspond to stabilizer generators. Parafermion stabilizer codes thus lead to multitude of models generalizing the Kitaev's one-dimensional (1D) chain of unpaired Majorana zero modes to higher dimensions ($D>2$) and to arbitrary interactions defined by the choice of stabilizer generators. An important question arising here is related to finite temperature stability of topological order in such systems. In general, $2$-dimensional lattice with local interactions cannot lead to stable topological order at finite temperature. Thus, it could be plausible to assume that by requiring some of the logical operators to be parity violating operators one can add additional protection to topological order where this additional protection relies on superselection rules. Whether such constructions can lead to the absence of parity conserving string-like logical operators (e.g. string-like logical operators are absent in the 
Haah's code \cite{Haah:PRA2011}) is an open problem.

\section*{Acknowledgments}
We are grateful to L. Pryadko, K. Shtengel, and S. Bravyi for multiple helpful discussions.
This work was supported in part by the NSF under Grants No.~Phy-1415600 and NSF-EPSCoR 1004094.

\bibliographystyle{apsrev}
\bibliography{parafermion,extra}

\begin{thebibliography}{68}
\expandafter\ifx\csname natexlab\endcsname\relax\def\natexlab#1{#1}\fi
\expandafter\ifx\csname bibnamefont\endcsname\relax
  \def\bibnamefont#1{#1}\fi
\expandafter\ifx\csname bibfnamefont\endcsname\relax
  \def\bibfnamefont#1{#1}\fi
\expandafter\ifx\csname citenamefont\endcsname\relax
  \def\citenamefont#1{#1}\fi
\expandafter\ifx\csname url\endcsname\relax
  \def\url#1{\texttt{#1}}\fi
\expandafter\ifx\csname urlprefix\endcsname\relax\def\urlprefix{URL }\fi
\providecommand{\bibinfo}[2]{#2}
\providecommand{\eprint}[2][]{\url{#2}}

\bibitem[{\citenamefont{Kitaev}(2003)}]{Kitaev2003}
\bibinfo{author}{\bibfnamefont{A.}~\bibnamefont{Kitaev}},
  \bibinfo{journal}{Ann. Phys. (N. Y).} \textbf{\bibinfo{volume}{303}},
  \bibinfo{pages}{2} (\bibinfo{year}{2003}).

\bibitem[{\citenamefont{Nayak et~al.}(2008)\citenamefont{Nayak, Simon, Stern,
  Freedman, and Sarma}}]{Nayak.Simon.ea:RoMP2008}
\bibinfo{author}{\bibfnamefont{C.}~\bibnamefont{Nayak}},
  \bibinfo{author}{\bibfnamefont{S.~H.} \bibnamefont{Simon}},
  \bibinfo{author}{\bibfnamefont{A.}~\bibnamefont{Stern}},
  \bibinfo{author}{\bibfnamefont{M.}~\bibnamefont{Freedman}}, \bibnamefont{and}
  \bibinfo{author}{\bibfnamefont{S.~D.} \bibnamefont{Sarma}},
  \bibinfo{journal}{Rev. Mod. Phys.} \textbf{\bibinfo{volume}{80}},
  \bibinfo{pages}{1083} (\bibinfo{year}{2008}).

\bibitem[{\citenamefont{Moore and Seiberg}(1989)}]{Moore.Seiberg:CiMP1989}
\bibinfo{author}{\bibfnamefont{G.}~\bibnamefont{Moore}} \bibnamefont{and}
  \bibinfo{author}{\bibfnamefont{N.}~\bibnamefont{Seiberg}},
  \bibinfo{journal}{Commun. in Math. Phys.} \textbf{\bibinfo{volume}{123}},
  \bibinfo{pages}{177} (\bibinfo{year}{1989}).

\bibitem[{\citenamefont{Witten}(1989)}]{Witten:CiMP1989}
\bibinfo{author}{\bibfnamefont{E.}~\bibnamefont{Witten}},
  \bibinfo{journal}{Commun. in Math. Phys.} \textbf{\bibinfo{volume}{121}},
  \bibinfo{pages}{351} (\bibinfo{year}{1989}).

\bibitem[{\citenamefont{Kitaev}(2001)}]{Kitaev2001}
\bibinfo{author}{\bibfnamefont{A.~Y.} \bibnamefont{Kitaev}},
  \bibinfo{journal}{Physics-Uspekhi} \textbf{\bibinfo{volume}{44}},
  \bibinfo{pages}{131} (\bibinfo{year}{2001}).

\bibitem[{\citenamefont{{Alicea} et~al.}(2011)\citenamefont{{Alicea}, {Oreg},
  {Refael}, {von Oppen}, and {Fisher}}}]{Alicea.Oreg.ea:NP2011}
\bibinfo{author}{\bibfnamefont{J.}~\bibnamefont{{Alicea}}},
  \bibinfo{author}{\bibfnamefont{Y.}~\bibnamefont{{Oreg}}},
  \bibinfo{author}{\bibfnamefont{G.}~\bibnamefont{{Refael}}},
  \bibinfo{author}{\bibfnamefont{F.}~\bibnamefont{{von Oppen}}},
  \bibnamefont{and} \bibinfo{author}{\bibfnamefont{M.~P.~A.}
  \bibnamefont{{Fisher}}}, \bibinfo{journal}{Nat. Phys.}
  \textbf{\bibinfo{volume}{7}}, \bibinfo{pages}{412} (\bibinfo{year}{2011}).

\bibitem[{\citenamefont{{Clarke} et~al.}(2011)\citenamefont{{Clarke}, {Sau},
  and {Tewari}}}]{Clarke.Sau.ea:2011}
\bibinfo{author}{\bibfnamefont{D.~J.} \bibnamefont{{Clarke}}},
  \bibinfo{author}{\bibfnamefont{J.~D.} \bibnamefont{{Sau}}}, \bibnamefont{and}
  \bibinfo{author}{\bibfnamefont{S.}~\bibnamefont{{Tewari}}},
  \bibinfo{journal}{Phys. Rev. B} \textbf{\bibinfo{volume}{84}},
  \bibinfo{eid}{035120} (\bibinfo{year}{2011}).

\bibitem[{\citenamefont{Fendley}(2012)}]{Fendley2012}
\bibinfo{author}{\bibfnamefont{P.}~\bibnamefont{Fendley}}, \bibinfo{journal}{J.
  Stat. Mech.} \textbf{\bibinfo{volume}{2012}}, \bibinfo{pages}{11020}
  (\bibinfo{year}{2012}).

\bibitem[{\citenamefont{{Barkeshli} and {Qi}}(2012)}]{Barkeshli.Qi:PRX2012}
\bibinfo{author}{\bibfnamefont{M.}~\bibnamefont{{Barkeshli}}} \bibnamefont{and}
  \bibinfo{author}{\bibfnamefont{X.-L.} \bibnamefont{{Qi}}},
  \bibinfo{journal}{Phys. Rev. X} \textbf{\bibinfo{volume}{2}},
  \bibinfo{eid}{031013} (\bibinfo{year}{2012}).

\bibitem[{\citenamefont{{Lindner} et~al.}(2012)\citenamefont{{Lindner}, {Berg},
  {Refael}, and {Stern}}}]{Lindner.Berg.ea:PRX2012}
\bibinfo{author}{\bibfnamefont{N.~H.} \bibnamefont{{Lindner}}},
  \bibinfo{author}{\bibfnamefont{E.}~\bibnamefont{{Berg}}},
  \bibinfo{author}{\bibfnamefont{G.}~\bibnamefont{{Refael}}}, \bibnamefont{and}
  \bibinfo{author}{\bibfnamefont{A.}~\bibnamefont{{Stern}}},
  \bibinfo{journal}{Phys. Rev. X} \textbf{\bibinfo{volume}{2}},
  \bibinfo{eid}{041002} (\bibinfo{year}{2012}).

\bibitem[{\citenamefont{Clarke et~al.}(2013)\citenamefont{Clarke, Alicea, and
  Shtengel}}]{Clarke2013}
\bibinfo{author}{\bibfnamefont{D.~J.} \bibnamefont{Clarke}},
  \bibinfo{author}{\bibfnamefont{J.}~\bibnamefont{Alicea}}, \bibnamefont{and}
  \bibinfo{author}{\bibfnamefont{K.}~\bibnamefont{Shtengel}},
  \bibinfo{journal}{Nat. Commun.} \textbf{\bibinfo{volume}{4}},
  \bibinfo{pages}{1348} (\bibinfo{year}{2013}).

\bibitem[{\citenamefont{{Cheng}}(2012)}]{Cheng:2012}
\bibinfo{author}{\bibfnamefont{M.}~\bibnamefont{{Cheng}}},
  \bibinfo{journal}{Phys. Rev. B} \textbf{\bibinfo{volume}{86}},
  \bibinfo{eid}{195126} (\bibinfo{year}{2012}).

\bibitem[{\citenamefont{{Vaezi}}(2013)}]{Vaezi:2013}
\bibinfo{author}{\bibfnamefont{A.}~\bibnamefont{{Vaezi}}},
  \bibinfo{journal}{Phys. Rev. B} \textbf{\bibinfo{volume}{87}},
  \bibinfo{eid}{035132} (\bibinfo{year}{2013}).

\bibitem[{\citenamefont{{Barkeshli}
  et~al.}(2013{\natexlab{a}})\citenamefont{{Barkeshli}, {Jian}, and
  {Qi}}}]{Barkeshli.Jian.ea:2013}
\bibinfo{author}{\bibfnamefont{M.}~\bibnamefont{{Barkeshli}}},
  \bibinfo{author}{\bibfnamefont{C.-M.} \bibnamefont{{Jian}}},
  \bibnamefont{and} \bibinfo{author}{\bibfnamefont{X.-L.} \bibnamefont{{Qi}}},
  \bibinfo{journal}{Phys. Rev. B} \textbf{\bibinfo{volume}{87}},
  \bibinfo{eid}{045130} (\bibinfo{year}{2013}{\natexlab{a}}).

\bibitem[{\citenamefont{{Barkeshli}
  et~al.}(2013{\natexlab{b}})\citenamefont{{Barkeshli}, {Jian}, and
  {Qi}}}]{Barkeshli.Jian.ea:2013b}
\bibinfo{author}{\bibfnamefont{M.}~\bibnamefont{{Barkeshli}}},
  \bibinfo{author}{\bibfnamefont{C.-M.} \bibnamefont{{Jian}}},
  \bibnamefont{and} \bibinfo{author}{\bibfnamefont{X.-L.} \bibnamefont{{Qi}}},
  \bibinfo{journal}{Phys. Rev. B} \textbf{\bibinfo{volume}{88}},
  \bibinfo{eid}{235103} (\bibinfo{year}{2013}{\natexlab{b}}).

\bibitem[{\citenamefont{{Hastings} et~al.}(2013)\citenamefont{{Hastings},
  {Nayak}, and {Wang}}}]{Hastings.Nayak.ea:2013}
\bibinfo{author}{\bibfnamefont{M.~B.} \bibnamefont{{Hastings}}},
  \bibinfo{author}{\bibfnamefont{C.}~\bibnamefont{{Nayak}}}, \bibnamefont{and}
  \bibinfo{author}{\bibfnamefont{Z.}~\bibnamefont{{Wang}}},
  \bibinfo{journal}{Phys. Rev. B} \textbf{\bibinfo{volume}{87}},
  \bibinfo{eid}{165421} (\bibinfo{year}{2013}).

\bibitem[{\citenamefont{{Oreg} et~al.}(2014)\citenamefont{{Oreg}, {Sela}, and
  {Stern}}}]{Oreg.Sela.ea:2014}
\bibinfo{author}{\bibfnamefont{Y.}~\bibnamefont{{Oreg}}},
  \bibinfo{author}{\bibfnamefont{E.}~\bibnamefont{{Sela}}}, \bibnamefont{and}
  \bibinfo{author}{\bibfnamefont{A.}~\bibnamefont{{Stern}}},
  \bibinfo{journal}{Phys. Rev. B} \textbf{\bibinfo{volume}{89}},
  \bibinfo{eid}{115402} (\bibinfo{year}{2014}).

\bibitem[{\citenamefont{{Burrello} et~al.}(2013)\citenamefont{{Burrello}, {van
  Heck}, and {Cobanera}}}]{Burrello.vanHeck.ea:2013}
\bibinfo{author}{\bibfnamefont{M.}~\bibnamefont{{Burrello}}},
  \bibinfo{author}{\bibfnamefont{B.}~\bibnamefont{{van Heck}}},
  \bibnamefont{and}
  \bibinfo{author}{\bibfnamefont{E.}~\bibnamefont{{Cobanera}}},
  \bibinfo{journal}{Phys. Rev. B} \textbf{\bibinfo{volume}{87}},
  \bibinfo{eid}{195422} (\bibinfo{year}{2013}).

\bibitem[{\citenamefont{{Mong} et~al.}(2014)\citenamefont{{Mong}, {Clarke},
  {Alicea}, {Lindner}, {Fendley}, {Nayak}, {Oreg}, {Stern}, {Berg}, {Shtengel}
  et~al.}}]{Mong.Clarke.ea:PRX2014}
\bibinfo{author}{\bibfnamefont{R.~S.~K.} \bibnamefont{{Mong}}},
  \bibinfo{author}{\bibfnamefont{D.~J.} \bibnamefont{{Clarke}}},
  \bibinfo{author}{\bibfnamefont{J.}~\bibnamefont{{Alicea}}},
  \bibinfo{author}{\bibfnamefont{N.~H.} \bibnamefont{{Lindner}}},
  \bibinfo{author}{\bibfnamefont{P.}~\bibnamefont{{Fendley}}},
  \bibinfo{author}{\bibfnamefont{C.}~\bibnamefont{{Nayak}}},
  \bibinfo{author}{\bibfnamefont{Y.}~\bibnamefont{{Oreg}}},
  \bibinfo{author}{\bibfnamefont{A.}~\bibnamefont{{Stern}}},
  \bibinfo{author}{\bibfnamefont{E.}~\bibnamefont{{Berg}}},
  \bibinfo{author}{\bibfnamefont{K.}~\bibnamefont{{Shtengel}}},
  \bibnamefont{et~al.}, \bibinfo{journal}{Phys. Rev. X}
  \textbf{\bibinfo{volume}{4}}, \bibinfo{eid}{011036} (\bibinfo{year}{2014}).

\bibitem[{\citenamefont{Klinovaja and Loss}(2014)}]{Klinovaja2014}
\bibinfo{author}{\bibfnamefont{J.}~\bibnamefont{Klinovaja}} \bibnamefont{and}
  \bibinfo{author}{\bibfnamefont{D.}~\bibnamefont{Loss}},
  \bibinfo{journal}{Phys. Rev. Lett.} \textbf{\bibinfo{volume}{112}},
  \bibinfo{pages}{246403} (\bibinfo{year}{2014}).

\bibitem[{\citenamefont{Tsvelik}(2014)}]{Tsvelik2014}
\bibinfo{author}{\bibfnamefont{A.~M.} \bibnamefont{Tsvelik}},
  \bibinfo{journal}{Phys. Rev. Lett.} \textbf{\bibinfo{volume}{113}},
  \bibinfo{pages}{066401} (\bibinfo{year}{2014}).

\bibitem[{\citenamefont{Ortiz et~al.}(2012)\citenamefont{Ortiz, Cobanera, and
  Nussinov}}]{Ortiz2012}
\bibinfo{author}{\bibfnamefont{G.}~\bibnamefont{Ortiz}},
  \bibinfo{author}{\bibfnamefont{E.}~\bibnamefont{Cobanera}}, \bibnamefont{and}
  \bibinfo{author}{\bibfnamefont{Z.}~\bibnamefont{Nussinov}},
  \bibinfo{journal}{Nucl. Phys. B} \textbf{\bibinfo{volume}{854}},
  \bibinfo{pages}{780} (\bibinfo{year}{2012}).

\bibitem[{\citenamefont{Nussinov and Ortiz}(2008)}]{Nussinov2008}
\bibinfo{author}{\bibfnamefont{Z.}~\bibnamefont{Nussinov}} \bibnamefont{and}
  \bibinfo{author}{\bibfnamefont{G.}~\bibnamefont{Ortiz}},
  \bibinfo{journal}{Phys. Rev. B} \textbf{\bibinfo{volume}{77}},
  \bibinfo{pages}{064302} (\bibinfo{year}{2008}).

\bibitem[{\citenamefont{Schulz et~al.}(2012)\citenamefont{Schulz, Dusuel,
  Or{\'u}s, Vidal, and Schmidt}}]{Schulz2012}
\bibinfo{author}{\bibfnamefont{M.~D.} \bibnamefont{Schulz}},
  \bibinfo{author}{\bibfnamefont{S.}~\bibnamefont{Dusuel}},
  \bibinfo{author}{\bibfnamefont{R.}~\bibnamefont{Or{\'u}s}},
  \bibinfo{author}{\bibfnamefont{J.}~\bibnamefont{Vidal}}, \bibnamefont{and}
  \bibinfo{author}{\bibfnamefont{K.~P.} \bibnamefont{Schmidt}},
  \bibinfo{journal}{New J. Phys.} \textbf{\bibinfo{volume}{14}},
  \bibinfo{pages}{025005} (\bibinfo{year}{2012}).

\bibitem[{\citenamefont{Bullock and Brennen}(2007)}]{Bullock2007}
\bibinfo{author}{\bibfnamefont{S.~S.} \bibnamefont{Bullock}} \bibnamefont{and}
  \bibinfo{author}{\bibfnamefont{G.~K.} \bibnamefont{Brennen}},
  \bibinfo{journal}{J. Phys. A Math. Theor.} \textbf{\bibinfo{volume}{40}},
  \bibinfo{pages}{3481} (\bibinfo{year}{2007}).

\bibitem[{\citenamefont{Vaezi}(2014)}]{Vaezi2014}
\bibinfo{author}{\bibfnamefont{A.}~\bibnamefont{Vaezi}},
  \bibinfo{journal}{Phys. Rev. X} \textbf{\bibinfo{volume}{4}},
  \bibinfo{pages}{031009} (\bibinfo{year}{2014}).

\bibitem[{\citenamefont{Nigg et~al.}(2014)\citenamefont{Nigg, M{\"u}ller,
  Martinez, Schindler, Hennrich, Monz, Martin-Delgado, and Blatt}}]{Nigg2014}
\bibinfo{author}{\bibfnamefont{D.}~\bibnamefont{Nigg}},
  \bibinfo{author}{\bibfnamefont{M.}~\bibnamefont{M{\"u}ller}},
  \bibinfo{author}{\bibfnamefont{E.}~\bibnamefont{Martinez}},
  \bibinfo{author}{\bibfnamefont{P.}~\bibnamefont{Schindler}},
  \bibinfo{author}{\bibfnamefont{M.}~\bibnamefont{Hennrich}},
  \bibinfo{author}{\bibfnamefont{T.}~\bibnamefont{Monz}},
  \bibinfo{author}{\bibfnamefont{M.}~\bibnamefont{Martin-Delgado}},
  \bibnamefont{and} \bibinfo{author}{\bibfnamefont{R.}~\bibnamefont{Blatt}},
  \bibinfo{journal}{Science} \textbf{\bibinfo{volume}{345}},
  \bibinfo{pages}{302} (\bibinfo{year}{2014}).

\bibitem[{\citenamefont{Bondesan and Quella}(2013)}]{Bondesan2013}
\bibinfo{author}{\bibfnamefont{R.}~\bibnamefont{Bondesan}} \bibnamefont{and}
  \bibinfo{author}{\bibfnamefont{T.}~\bibnamefont{Quella}},
  \bibinfo{journal}{J. Stat. Mech.} \textbf{\bibinfo{volume}{2013}},
  \bibinfo{pages}{P10024} (\bibinfo{year}{2013}).

\bibitem[{\citenamefont{Motruk et~al.}(2013)\citenamefont{Motruk, Berg, Turner,
  and Pollmann}}]{Motruk2013}
\bibinfo{author}{\bibfnamefont{J.}~\bibnamefont{Motruk}},
  \bibinfo{author}{\bibfnamefont{E.}~\bibnamefont{Berg}},
  \bibinfo{author}{\bibfnamefont{A.~M.} \bibnamefont{Turner}},
  \bibnamefont{and} \bibinfo{author}{\bibfnamefont{F.}~\bibnamefont{Pollmann}},
  \bibinfo{journal}{Phys. Rev. B} \textbf{\bibinfo{volume}{88}},
  \bibinfo{pages}{085115} (\bibinfo{year}{2013}).

\bibitem[{\citenamefont{Gottesman}(1997)}]{Gottesman1997}
\bibinfo{author}{\bibfnamefont{D.}~\bibnamefont{Gottesman}}, Ph.D. thesis,
  \bibinfo{school}{Caltech} (\bibinfo{year}{1997}).

\bibitem[{\citenamefont{Dennis et~al.}(2002)\citenamefont{Dennis, Kitaev,
  Landahl, and Preskill}}]{Dennis2002}
\bibinfo{author}{\bibfnamefont{E.}~\bibnamefont{Dennis}},
  \bibinfo{author}{\bibfnamefont{A.}~\bibnamefont{Kitaev}},
  \bibinfo{author}{\bibfnamefont{A.}~\bibnamefont{Landahl}}, \bibnamefont{and}
  \bibinfo{author}{\bibfnamefont{J.}~\bibnamefont{Preskill}},
  \bibinfo{journal}{J. Math. Phys.} \textbf{\bibinfo{volume}{43}},
  \bibinfo{pages}{4452} (\bibinfo{year}{2002}).

\bibitem[{\citenamefont{{Bacon}}(2006)}]{Bacon:2006}
\bibinfo{author}{\bibfnamefont{D.}~\bibnamefont{{Bacon}}},
  \bibinfo{journal}{Phys. Rev. A} \textbf{\bibinfo{volume}{73}},
  \bibinfo{eid}{012340} (\bibinfo{year}{2006}).

\bibitem[{\citenamefont{{Bravyi} and {Terhal}}(2009)}]{Bravyi.Terhal:NJoP2009}
\bibinfo{author}{\bibfnamefont{S.}~\bibnamefont{{Bravyi}}} \bibnamefont{and}
  \bibinfo{author}{\bibfnamefont{B.}~\bibnamefont{{Terhal}}},
  \bibinfo{journal}{New J. Phys.} \textbf{\bibinfo{volume}{11}},
  \bibinfo{eid}{043029} (\bibinfo{year}{2009}).

\bibitem[{\citenamefont{{Landon-Cardinal} and
  {Poulin}}(2013)}]{Landon-Cardinal.Poulin:PRL2013}
\bibinfo{author}{\bibfnamefont{O.}~\bibnamefont{{Landon-Cardinal}}}
  \bibnamefont{and} \bibinfo{author}{\bibfnamefont{D.}~\bibnamefont{{Poulin}}},
  \bibinfo{journal}{Phys. Rev. Lett.} \textbf{\bibinfo{volume}{110}},
  \bibinfo{eid}{090502} (\bibinfo{year}{2013}).

\bibitem[{\citenamefont{{Sau} et~al.}(2010)\citenamefont{{Sau}, {Tewari}, and
  {Das Sarma}}}]{Sau.Tewari.ea:2010}
\bibinfo{author}{\bibfnamefont{J.~D.} \bibnamefont{{Sau}}},
  \bibinfo{author}{\bibfnamefont{S.}~\bibnamefont{{Tewari}}}, \bibnamefont{and}
  \bibinfo{author}{\bibfnamefont{S.}~\bibnamefont{{Das Sarma}}},
  \bibinfo{journal}{Phys. Rev. A} \textbf{\bibinfo{volume}{82}},
  \bibinfo{eid}{052322} (\bibinfo{year}{2010}).

\bibitem[{\citenamefont{{Hassler} et~al.}(2010)\citenamefont{{Hassler},
  {Akhmerov}, {Hou}, and {Beenakker}}}]{Hassler.Akhmerov.ea:NJoP2010}
\bibinfo{author}{\bibfnamefont{F.}~\bibnamefont{{Hassler}}},
  \bibinfo{author}{\bibfnamefont{A.~R.} \bibnamefont{{Akhmerov}}},
  \bibinfo{author}{\bibfnamefont{C.-Y.} \bibnamefont{{Hou}}}, \bibnamefont{and}
  \bibinfo{author}{\bibfnamefont{C.~W.~J.} \bibnamefont{{Beenakker}}},
  \bibinfo{journal}{New J. Phys.} \textbf{\bibinfo{volume}{12}},
  \bibinfo{eid}{125002} (\bibinfo{year}{2010}).

\bibitem[{\citenamefont{{Jiang} et~al.}(2011)\citenamefont{{Jiang}, {Kane}, and
  {Preskill}}}]{Jiang.Kane.ea:PRL2011}
\bibinfo{author}{\bibfnamefont{L.}~\bibnamefont{{Jiang}}},
  \bibinfo{author}{\bibfnamefont{C.~L.} \bibnamefont{{Kane}}},
  \bibnamefont{and}
  \bibinfo{author}{\bibfnamefont{J.}~\bibnamefont{{Preskill}}},
  \bibinfo{journal}{Phys. Rev. Lett.} \textbf{\bibinfo{volume}{106}},
  \bibinfo{eid}{130504} (\bibinfo{year}{2011}).

\bibitem[{\citenamefont{{Shor}}(1995)}]{Shor:1995}
\bibinfo{author}{\bibfnamefont{P.~W.} \bibnamefont{{Shor}}},
  \bibinfo{journal}{Phys. Rev. A} \textbf{\bibinfo{volume}{52}},
  \bibinfo{pages}{2493} (\bibinfo{year}{1995}).

\bibitem[{\citenamefont{{Knill} and {Laflamme}}(1997)}]{Knill.Laflamme:1997}
\bibinfo{author}{\bibfnamefont{E.}~\bibnamefont{{Knill}}} \bibnamefont{and}
  \bibinfo{author}{\bibfnamefont{R.}~\bibnamefont{{Laflamme}}},
  \bibinfo{journal}{Phys. Rev. A} \textbf{\bibinfo{volume}{55}},
  \bibinfo{pages}{900} (\bibinfo{year}{1997}).

\bibitem[{\citenamefont{{Steane}}(1996)}]{Steane:PRL1996}
\bibinfo{author}{\bibfnamefont{A.~M.} \bibnamefont{{Steane}}},
  \bibinfo{journal}{Phys. Rev. Lett.} \textbf{\bibinfo{volume}{77}},
  \bibinfo{pages}{793} (\bibinfo{year}{1996}).

\bibitem[{\citenamefont{Rains}(1999)}]{Rains:ITITo1999}
\bibinfo{author}{\bibfnamefont{E.~M.} \bibnamefont{Rains}},
  \bibinfo{journal}{IEEE Trans. Inf. Theory} \textbf{\bibinfo{volume}{45}},
  \bibinfo{pages}{1827} (\bibinfo{year}{1999}).

\bibitem[{\citenamefont{Ashikhmin and Knill}(2001)}]{Ashikhmin.Knill:ITIT2001}
\bibinfo{author}{\bibfnamefont{A.}~\bibnamefont{Ashikhmin}} \bibnamefont{and}
  \bibinfo{author}{\bibfnamefont{E.}~\bibnamefont{Knill}},
  \bibinfo{journal}{IEEE Trans. Inform. Theory} \textbf{\bibinfo{volume}{47}},
  \bibinfo{pages}{3065} (\bibinfo{year}{2001}).

\bibitem[{\citenamefont{Schlingemann and
  Werner}(2001)}]{Schlingemann.Werner:2002}
\bibinfo{author}{\bibfnamefont{D.}~\bibnamefont{Schlingemann}}
  \bibnamefont{and} \bibinfo{author}{\bibfnamefont{R.~F.}
  \bibnamefont{Werner}}, \bibinfo{journal}{Phys. Rev. A}
  \textbf{\bibinfo{volume}{65}}, \bibinfo{eid}{012308} (\bibinfo{year}{2001}).

\bibitem[{\citenamefont{Grassl et~al.}(2004)\citenamefont{Grassl, Beth, and
  Rã–tteler}}]{GRASSL.BETH.ea:IJoQI2004}
\bibinfo{author}{\bibfnamefont{M.}~\bibnamefont{Grassl}},
  \bibinfo{author}{\bibfnamefont{T.}~\bibnamefont{Beth}}, \bibnamefont{and}
  \bibinfo{author}{\bibfnamefont{M.}~\bibnamefont{Rã–tteler}},
  \bibinfo{journal}{Int. J. Quantum Inf.} \textbf{\bibinfo{volume}{02}},
  \bibinfo{pages}{55} (\bibinfo{year}{2004}).

\bibitem[{\citenamefont{Looi et~al.}(2008)\citenamefont{Looi, Yu, Gheorghiu,
  and Griffiths}}]{Looi.Yu.ea:PRA2008}
\bibinfo{author}{\bibfnamefont{S.~Y.} \bibnamefont{Looi}},
  \bibinfo{author}{\bibfnamefont{L.}~\bibnamefont{Yu}},
  \bibinfo{author}{\bibfnamefont{V.}~\bibnamefont{Gheorghiu}},
  \bibnamefont{and} \bibinfo{author}{\bibfnamefont{R.~B.}
  \bibnamefont{Griffiths}}, \bibinfo{journal}{Phys. Rev. A}
  \textbf{\bibinfo{volume}{78}}, \bibinfo{pages}{042303}
  (\bibinfo{year}{2008}).

\bibitem[{\citenamefont{{Hu} et~al.}(2008)\citenamefont{{Hu}, {Tang}, {Zhao},
  {Chen}, {Yu}, and {Oh}}}]{Hu.Tang.ea:2008}
\bibinfo{author}{\bibfnamefont{D.}~\bibnamefont{{Hu}}},
  \bibinfo{author}{\bibfnamefont{W.}~\bibnamefont{{Tang}}},
  \bibinfo{author}{\bibfnamefont{M.}~\bibnamefont{{Zhao}}},
  \bibinfo{author}{\bibfnamefont{Q.}~\bibnamefont{{Chen}}},
  \bibinfo{author}{\bibfnamefont{S.}~\bibnamefont{{Yu}}}, \bibnamefont{and}
  \bibinfo{author}{\bibfnamefont{C.~H.} \bibnamefont{{Oh}}},
  \bibinfo{journal}{Phys. Rev. A} \textbf{\bibinfo{volume}{78}},
  \bibinfo{eid}{012306} (\bibinfo{year}{2008}).

\bibitem[{\citenamefont{Gheorghiu et~al.}(2010)\citenamefont{Gheorghiu, Looi,
  and Griffiths}}]{Gheorghiu2010}
\bibinfo{author}{\bibfnamefont{V.}~\bibnamefont{Gheorghiu}},
  \bibinfo{author}{\bibfnamefont{S.~Y.} \bibnamefont{Looi}}, \bibnamefont{and}
  \bibinfo{author}{\bibfnamefont{R.~B.} \bibnamefont{Griffiths}},
  \bibinfo{journal}{Phys. Rev. A} \textbf{\bibinfo{volume}{81}},
  \bibinfo{pages}{032326} (\bibinfo{year}{2010}).

\bibitem[{\citenamefont{Ketkar et~al.}(2006)\citenamefont{Ketkar, Klappenecker,
  Kumar, and Sarvepalli}}]{Ketkar2006}
\bibinfo{author}{\bibfnamefont{A.}~\bibnamefont{Ketkar}},
  \bibinfo{author}{\bibfnamefont{A.}~\bibnamefont{Klappenecker}},
  \bibinfo{author}{\bibfnamefont{S.}~\bibnamefont{Kumar}}, \bibnamefont{and}
  \bibinfo{author}{\bibfnamefont{P.}~\bibnamefont{Sarvepalli}},
  \bibinfo{journal}{IEEE Trans. Inf. Theory} \textbf{\bibinfo{volume}{52}},
  \bibinfo{pages}{4892} (\bibinfo{year}{2006}).

\bibitem[{\citenamefont{Chen et~al.}(2008)\citenamefont{Chen, Zeng, and
  Chuang}}]{Chen.Zeng.ea:PRA2008}
\bibinfo{author}{\bibfnamefont{X.}~\bibnamefont{Chen}},
  \bibinfo{author}{\bibfnamefont{B.}~\bibnamefont{Zeng}}, \bibnamefont{and}
  \bibinfo{author}{\bibfnamefont{I.~L.} \bibnamefont{Chuang}},
  \bibinfo{journal}{Phys. Rev. A} \textbf{\bibinfo{volume}{78}},
  \bibinfo{pages}{062315} (\bibinfo{year}{2008}).

\bibitem[{\citenamefont{Gheorghiu}(2014)}]{Gheorghiu2014}
\bibinfo{author}{\bibfnamefont{V.}~\bibnamefont{Gheorghiu}},
  \bibinfo{journal}{Phys. Lett. A} \textbf{\bibinfo{volume}{378}},
  \bibinfo{pages}{505} (\bibinfo{year}{2014}).

\bibitem[{\citenamefont{Bravyi et~al.}(2010)\citenamefont{Bravyi, Terhal, and
  Leemhuis}}]{Bravyi2010}
\bibinfo{author}{\bibfnamefont{S.}~\bibnamefont{Bravyi}},
  \bibinfo{author}{\bibfnamefont{B.~M.} \bibnamefont{Terhal}},
  \bibnamefont{and} \bibinfo{author}{\bibfnamefont{B.}~\bibnamefont{Leemhuis}},
  \bibinfo{journal}{New J. Phys.} \textbf{\bibinfo{volume}{12}},
  \bibinfo{pages}{083039} (\bibinfo{year}{2010}).

\bibitem[{\citenamefont{{Rainis} and {Loss}}(2012)}]{Rainis:may2012}
\bibinfo{author}{\bibfnamefont{D.}~\bibnamefont{{Rainis}}} \bibnamefont{and}
  \bibinfo{author}{\bibfnamefont{D.}~\bibnamefont{{Loss}}},
  \bibinfo{journal}{Phys. Rev. B} \textbf{\bibinfo{volume}{85}},
  \bibinfo{eid}{174533} (\bibinfo{year}{2012}).

\bibitem[{\citenamefont{{Burnell} et~al.}(2013)\citenamefont{{Burnell},
  {Shnirman}, and {Oreg}}}]{Burnell:dec2013}
\bibinfo{author}{\bibfnamefont{F.~J.} \bibnamefont{{Burnell}}},
  \bibinfo{author}{\bibfnamefont{A.}~\bibnamefont{{Shnirman}}},
  \bibnamefont{and} \bibinfo{author}{\bibfnamefont{Y.}~\bibnamefont{{Oreg}}},
  \bibinfo{journal}{Phys. Rev. B} \textbf{\bibinfo{volume}{88}},
  \bibinfo{eid}{224507} (\bibinfo{year}{2013}).

\bibitem[{\citenamefont{Weyl}(1950)}]{weyl1950theory}
\bibinfo{author}{\bibfnamefont{H.}~\bibnamefont{Weyl}},
  \emph{\bibinfo{title}{The theory of groups and quantum mechanics}}
  (\bibinfo{publisher}{Courier Dover Publications}, \bibinfo{year}{1950}).

\bibitem[{\citenamefont{Schwinger and Englert}(2001)}]{schwinger2001quantum}
\bibinfo{author}{\bibfnamefont{J.}~\bibnamefont{Schwinger}} \bibnamefont{and}
  \bibinfo{author}{\bibfnamefont{B.-G.} \bibnamefont{Englert}},
  \emph{\bibinfo{title}{Quantum mechanics: symbolism of atomic measurements}}
  (\bibinfo{publisher}{Springer}, \bibinfo{year}{2001}).

\bibitem[{\citenamefont{Gottesman}(1996)}]{Gottesman1996}
\bibinfo{author}{\bibfnamefont{D.}~\bibnamefont{Gottesman}},
  \bibinfo{journal}{Phys. Rev. A} \textbf{\bibinfo{volume}{54}},
  \bibinfo{pages}{1862} (\bibinfo{year}{1996}).

\bibitem[{\citenamefont{Calderbank et~al.}(1998)\citenamefont{Calderbank,
  Rains, Shor, and Sloane}}]{Calderbank98quantumerror}
\bibinfo{author}{\bibfnamefont{A.~R.} \bibnamefont{Calderbank}},
  \bibinfo{author}{\bibfnamefont{E.~M.} \bibnamefont{Rains}},
  \bibinfo{author}{\bibfnamefont{P.~W.} \bibnamefont{Shor}}, \bibnamefont{and}
  \bibinfo{author}{\bibfnamefont{N.~J.~A.} \bibnamefont{Sloane}},
  \bibinfo{journal}{IEEE Trans. Inform. Theory} \textbf{\bibinfo{volume}{44}},
  \bibinfo{pages}{1369} (\bibinfo{year}{1998}).

\bibitem[{\citenamefont{{Calderbank} and {Shor}}(1996)}]{Calderbank.Shor:1996}
\bibinfo{author}{\bibfnamefont{A.~R.} \bibnamefont{{Calderbank}}}
  \bibnamefont{and} \bibinfo{author}{\bibfnamefont{P.~W.}
  \bibnamefont{{Shor}}}, \bibinfo{journal}{Phys. Rev. A}
  \textbf{\bibinfo{volume}{54}}, \bibinfo{pages}{1098} (\bibinfo{year}{1996}).

\bibitem[{Note1()}]{Note1}
Note1, \bibinfo{note}{exhaustive search takes exponential time in $D$, thus we
  were unable to examine $D>7$ cases and determine the optimal $D$. A better
  algorithm may allow determining this value.}

\bibitem[{\citenamefont{Viyuela et~al.}(2012)\citenamefont{Viyuela, Rivas, and
  Martin-Delgado}}]{Viyuela2012}
\bibinfo{author}{\bibfnamefont{O.}~\bibnamefont{Viyuela}},
  \bibinfo{author}{\bibfnamefont{A.}~\bibnamefont{Rivas}}, \bibnamefont{and}
  \bibinfo{author}{\bibfnamefont{M.~A.} \bibnamefont{Martin-Delgado}},
  \bibinfo{journal}{New J. Phys.} \textbf{\bibinfo{volume}{14}},
  \bibinfo{pages}{033044} (\bibinfo{year}{2012}).

\bibitem[{\citenamefont{Anwar et~al.}(2012)\citenamefont{Anwar, Campbell, and
  Browne}}]{Anwar2012}
\bibinfo{author}{\bibfnamefont{H.}~\bibnamefont{Anwar}},
  \bibinfo{author}{\bibfnamefont{E.~T.} \bibnamefont{Campbell}},
  \bibnamefont{and} \bibinfo{author}{\bibfnamefont{D.~E.}
  \bibnamefont{Browne}}, \bibinfo{journal}{New J. Phys.}
  \textbf{\bibinfo{volume}{14}}, \bibinfo{pages}{063006}
  (\bibinfo{year}{2012}).

\bibitem[{\citenamefont{Tillich and Zemor}(2009)}]{Tillich2009}
\bibinfo{author}{\bibfnamefont{J.-P.} \bibnamefont{Tillich}} \bibnamefont{and}
  \bibinfo{author}{\bibfnamefont{G.}~\bibnamefont{Zemor}}, in
  \emph{\bibinfo{booktitle}{Information Theory, 2009. ISIT 2009. IEEE
  International Symposium on}} (\bibinfo{year}{2009}), pp. \bibinfo{pages}{799
  --803}.

\bibitem[{\citenamefont{{Kovalev} and {Pryadko}}(2013)}]{Kovalev.Pryadko:2013}
\bibinfo{author}{\bibfnamefont{A.~A.} \bibnamefont{{Kovalev}}}
  \bibnamefont{and} \bibinfo{author}{\bibfnamefont{L.~P.}
  \bibnamefont{{Pryadko}}}, \bibinfo{journal}{Phys. Rev. A}
  \textbf{\bibinfo{volume}{88}}, \bibinfo{eid}{012311} (\bibinfo{year}{2013}).

\bibitem[{Note2()}]{Note2}
Note2, \bibinfo{note}{the presence of parity-violating operators does not
  prevent quantum computation. Kitaev chain contains parity-violating operators
  as well, nevertheless a topological qubit can be defined by using 4 Majorana
  edge modes or a pair of topological regions \cite
  {Fu2009,Lutchyn2010,Drummond2014}. Storage and manipulation of information
  takes place in the code space corresponding to a given parity sector of the
  Hilbert space.}

\bibitem[{\citenamefont{Haah}(2011)}]{Haah:PRA2011}
\bibinfo{author}{\bibfnamefont{J.}~\bibnamefont{Haah}}, \bibinfo{journal}{Phys.
  Rev. A} \textbf{\bibinfo{volume}{83}}, \bibinfo{pages}{042330}
  (\bibinfo{year}{2011}).

\bibitem[{\citenamefont{Fu and Kane}(2009)}]{Fu2009}
\bibinfo{author}{\bibfnamefont{L.}~\bibnamefont{Fu}} \bibnamefont{and}
  \bibinfo{author}{\bibfnamefont{C.~L.} \bibnamefont{Kane}},
  \bibinfo{journal}{Phys. Rev. B} \textbf{\bibinfo{volume}{79}},
  \bibinfo{pages}{161408} (\bibinfo{year}{2009}).

\bibitem[{\citenamefont{Lutchyn et~al.}(2010)\citenamefont{Lutchyn, Sau, and
  Das~Sarma}}]{Lutchyn2010}
\bibinfo{author}{\bibfnamefont{R.~M.} \bibnamefont{Lutchyn}},
  \bibinfo{author}{\bibfnamefont{J.~D.} \bibnamefont{Sau}}, \bibnamefont{and}
  \bibinfo{author}{\bibfnamefont{S.}~\bibnamefont{Das~Sarma}},
  \bibinfo{journal}{Phys. Rev. Lett.} \textbf{\bibinfo{volume}{105}},
  \bibinfo{pages}{077001} (\bibinfo{year}{2010}).

\bibitem[{\citenamefont{Drummond et~al.}(2014)\citenamefont{Drummond, Kovalev,
  Hou, Shtengel, and Pryadko}}]{Drummond2014}
\bibinfo{author}{\bibfnamefont{D.~E.} \bibnamefont{Drummond}},
  \bibinfo{author}{\bibfnamefont{A.~A.} \bibnamefont{Kovalev}},
  \bibinfo{author}{\bibfnamefont{C.-Y.} \bibnamefont{Hou}},
  \bibinfo{author}{\bibfnamefont{K.}~\bibnamefont{Shtengel}}, \bibnamefont{and}
  \bibinfo{author}{\bibfnamefont{L.~P.} \bibnamefont{Pryadko}},
  \bibinfo{journal}{Phys. Rev. B} \textbf{\bibinfo{volume}{90}},
  \bibinfo{pages}{115404} (\bibinfo{year}{2014}).

\end{thebibliography}

\end{document}